\documentclass[11pt,a4paper]{article}

  \usepackage{amsmath,amsfonts}
  \usepackage{latexsym,amssymb}
  \usepackage[normalem]{ulem}
  \usepackage{amsthm}
  \usepackage{fullpage}
  \usepackage{tikz}
  \usepackage[T1]{fontenc}
  \usepackage[utf8]{inputenc}
  \usepackage{comment}
  \usepackage{xspace}
  \usepackage{enumerate}
  \usepackage{listings}
  \usepackage[hidelinks]{hyperref}
  \usepackage[noline,noend]{algorithm2e}
  \usetikzlibrary{decorations.pathreplacing,calc,snakes}
  \usepackage[]{authblk} 
  \usepackage{lmodern}

\newcommand{\defdsproblem}[2]{
\vspace{1mm}
\noindent\fbox{
\begin{minipage}{0.96\columnwidth}
  \textsc{#1} \\
#2

\end{minipage}
}
}

  \renewcommand{\epsilon}{\varepsilon}

  \newcommand{\barS}{\widetilde{S}}
  \newcommand{\barG}[1]{\widetilde{G}_{#1}}

  \newcommand{\MakeReverseFactorFree}{\textsl{\textsc{Make-Reverse-Factor-Free}}}
  \newcommand{\MakeFactorFree}{\textsl{\textsc{Make-Factor-Free}}}

  \DeclareMathOperator{\Pref}{\mathit{Pref}}
  \DeclareMathOperator{\PrefSet}{\mathit{PrefSet}}
  \DeclareMathOperator{\SufSet}{\mathit{SufSet}}

  \newcommand{\ov}{\mathit{ov}}
  \DeclareMathOperator{\pr}{\mathit{pr}}
  \newcommand{\OV}{\mathit{OV}}
  \DeclareMathOperator{\opt}{\mathrm{opt}}

  \newcommand{\SCS}{\mathit{SCS}}
  \newcommand{\SCSR}{\mathit{SCS\mbox{-}R}}
  \newcommand{\pgreedy}{\mathit{pgreedy}}
  \newcommand{\str}{\mathit{str}}
  \DeclareMathOperator{\Oh}{\mathcal{O}}
  
\theoremstyle{plain}
\newtheorem{theorem}{Theorem}
\newtheorem{proposition}[theorem]{Proposition}
\newtheorem{lemma}[theorem]{Lemma}

\theoremstyle{definition}
\newtheorem{example}{Example}

\theoremstyle{remark}

\newtheorem{claim}{Claim}

\newtheorem{observation}{Observation}  
  
\begin{document}

\title{On the Greedy Algorithm for the Shortest Common Superstring Problem with Reversals}

\author{
Gabriele Fici

\textit{Dipartimento di Matematica e Informatica, Universit\`a di Palermo, Italy}

\texttt{gabriele.fici@unipa.it}

\and Tomasz Kociumaka

\textit{Faculty of Mathematics, Informatics and Mechanics, University of Warsaw,  Poland}

\texttt{kociumaka@mimuw.edu.pl}

\and Jakub Radoszewski

\textit{Faculty of Mathematics, Informatics and Mechanics, University of Warsaw,  Poland and}

\textit{Newton International Fellow at Department~of Informatics, King's College London, UK}

\texttt{jrad@mimuw.edu.pl}

\and Wojciech Rytter

\textit{Faculty of Mathematics, Informatics and Mechanics, University of Warsaw,  Poland}

\texttt{rytter@mimuw.edu.pl}

\and Tomasz Wale\'{n}

\textit{Faculty of Mathematics, Informatics and Mechanics, University of Warsaw,  Poland}

\texttt{walen@mimuw.edu.pl}
}

\date{}

\sloppy  

\maketitle

  \begin{abstract}
We study a variation of the classical Shortest Common Superstring (SCS) problem in which a shortest superstring of a finite set of strings $S$ is sought containing as a factor every string of $S$ \emph{or} its reversal. We call this problem Shortest Common Superstring with Reversals (SCS-R). This problem has been introduced by Jiang et al. \cite{DBLP:journals/ipl/JiangLD92}, 
who designed a greedy-like algorithm with length approximation ratio $4$.
In this paper, we show that a natural adaptation of the classical greedy algorithm for SCS has (optimal) \textit{compression ratio} $\frac12$, i.e., the sum of the overlaps in the output string is at least half the sum of the overlaps in an optimal solution. 
We also provide a linear-time implementation of our algorithm.
\end{abstract}



\section{Introduction}

The Shortest Common Superstring (SCS) problem is a classical combinatorial problem on strings with applications in many domains, e.g.~DNA fragment assembly, data compression, etc. (see \cite{GePi14} for a recent survey). It consists, given a finite set of strings $S$ over an alphabet $\Sigma$, in finding a shortest string  containing as factors (substrings) all the strings in $S$. The decision version of the problem is known to be NP-complete \cite{MaSt77,GallantPhD,GaMaSt80}, even under several restrictions on the structure of $S$ (see again \cite{GePi14}). However, a particularly simple greedy algorithm introduced by Gallant in his Ph.D. thesis \cite{GallantPhD} is widely used in applications since it has very good performance in practice (see for instance \cite{Ma09} and references therein). It consists in repetitively replacing a pair of strings with maximum overlap with the string obtained by overlapping the two strings, until one string remains. The greedy algorithm can be implemented using Aho-Corasick automaton in $\Oh(n)$ randomized time (with hashing on the symbols of the alphabet) or $\Oh(n \min(\log m,\log |\Sigma|))$ deterministic time (see \cite{Uk90}), where  $n$ is the sum of the lengths of the strings in $S$ and $m$ its cardinality.

The approximation of the greedy algorithm is usually measured in two different ways: one consists in taking into account the \emph{approximation ratio} (also known as the \emph{length ratio}) $k_g/k_{min}$, where $k_g$ is the length of the output string of greedy and $k_{min}$ the length of a shortest superstring, the other consists in taking into account the \emph{compression ratio} $(n-k_g)/(n-k_{min})$. 

For the approximation ratio, Turner \cite{Tu89} proved that there is no constant $c<2$ such that $k_g/k_{min}\leq c$. The \emph{greedy conjecture} states that  this approximation ratio is in fact $2$ \cite{BlJiLiTrYa94}. The best bound currently known is 3.5 due to Kaplan and Shafrir \cite{KaSh05}. Algorithms with better approximation ratio are known; the best one is due to Mucha, with an approximation ratio of $2\frac{11}{23}$ \cite{Mu13}. 

For the compression ratio, Tarhio and Ukkonen \cite{TaUk88} proved that $(n-k_g)/(n-k_{min})\geq \frac12$ and this bound is tight, since it is achieved for the set $S=\{ab^h,b^ha,b^{h+1}\}$ when greedy makes the first choice merging the first two strings together.

Let us formally state the SCS problem:

\bigskip

  \defdsproblem{Shortest Common Superstring ($\SCS$)}{
    \textbf{Input:} strings $S=\{s_1,\ldots,s_m\}$ of total length $n$.\\
    \textbf{Output:} a shortest string $u$ that contains $s_i$ for each $i=1,\ldots,m$ as a factor.
  }

  \bigskip
  
    Several variations of SCS have been considered in literature. For example, shortest common superstring problem with reverse complements was considered in
  \cite{DBLP:journals/algorithmica/KececiogluM95}. In this setting the alphabet is $\Sigma=\{\mathtt{a,t,g,c}\}$ and the complement of a string $s$ is $\bar{s}^{R}$, where $\, \bar{}\, $ is defined by
  $\bar{\mathtt{a}}=\mathtt{t}$,
  $\bar{\mathtt{t}}=\mathtt{a}$,
  $\bar{\mathtt{g}}=\mathtt{c}$,
  $\bar{\mathtt{c}}=\mathtt{g}$,
  and $t^R$ denotes the reversal of $t$, that is the string obtained reading $t$ backwards.   In particular, this problem was shown to be NP-complete.
  
   Other variations of the SCS problem can be found in \cite{Jiang94,cpm,swap,rest}.   
   
    In this paper, we address the problem of searching for a string $u$ of minimal length such that for every $s_i\in S$,  $u$ contains as a factor $s_i$ \emph{or} its reversal $s_i^{R}$.
  
    \bigskip
    
  \defdsproblem{Shortest Common Superstring with Reversals ($\SCSR$)}{
    \textbf{Input:} strings $S=\{s_1,\ldots,s_m\}$ of total length $n$.\\
    \textbf{Output:} a shortest string $u$ that contains for each $i=1,\ldots,m$ at least one of the
    strings $s_i$ or $s^{R}_i$ as a factor.
  }
  
    \bigskip
For example, if $S=\{aabb, aaac, abbb\}$, then a solution of SCS-R for $S$ is $caaabbb$. Notice that a shortest superstring with reversals can be much shorter than a classical shortest superstring. An extremal example is given by an input set of the form $S=\{ab^h,cb^h\}$.

  The SCS-R problem was already considered by Jiang et al.\ \cite{DBLP:journals/ipl/JiangLD92}, who observed (not giving any proof) that the problem is still NP-hard. We provide a proof at the end of the paper.
  
  In \cite{DBLP:journals/ipl/JiangLD92}, the authors proposed a greedy 4-approximation algorithm. Here, we show that an adaptation of the classical greedy algorithm can be used for solving the SCS-R problem with an (optimal) compression ratio $\frac12$, and that this algorithm can be implemented in linear time with respect to the total size of the input set.

\section{Basics and Notation}

Let $\Sigma$ be a finite alphabet. We assume that $\Sigma$ is linearly sortable, e.g., $\Sigma=\{0,\ldots,n^{\Oh(1)}\}$.  The \emph{length} of a string $s$ over $\Sigma$ is denoted by $|s|$. The \emph{empty string}, denoted by $\epsilon$, is the unique string of length zero.
A string $t$ \emph{occurs} in a string $s$ if $s=vtz$ for some strings $v,z$. In this case we say that $t$ is a \emph{factor} of $s$. 
In particular, we say that $t$ is a \emph{prefix} of $s$ when $v=\epsilon$ and a \emph{suffix} of $s$ when $z=\epsilon$.
We say that a factor $t$ is \emph{proper} if $s \ne t$.

The string $s^{R}$ obtained by reading $s$ from right to left is called the \emph{reversal} (or \emph{mirror image}) of $s$.
Given a set of strings $S=\{s_1,\ldots,s_m\}$, we define the set $S^R = \{s_1^{R},\ldots,s_m^{R}\}$ and the set $\barS=S \cup S^R$.

Given two strings $u,v$, we define the (maximum) overlap between $u$ and $v$, denoted by $\ov(u,v)$, as the length of the longest suffix of $u$ that is also a prefix of $v$. Sometimes we abuse the notation and also say that the suffix of $u$ of length $\ov(u,v)$ is the overlap of $u$ and $v$.
In general $\ov(u,v)$ is not equal to $\ov(v,u)$, but it is readily verified that $\ov(u,v) = \ov(v^{R},u^{R})$.
Additionally, we define $\pr(u,v)$ as the prefix of $u$ obtained by removing the suffix of length $\ov(u,v)$ and denote $u \otimes v=\pr(u,v) v$.
Note that the $\otimes$ operation is in general neither symmetric nor associative.

A set of strings $S$ is called \emph{factor-free} if no string in $S$ is a factor of another string in $S$. We say that $S$ is \emph{reverse-factor-free} if there are no distinct strings $u,v\in S$ such that $u$ is a factor of $v$ or~$v^R$.

Given a factor-free set of strings  $S=\{s_1,\ldots,s_m\}$, the SCS problem for $S$ is known to be equivalent to that of finding a maximum-weight Hamiltonian path $\pi$ in the \emph{overlap graph} $G_S$, which is a directed weighted graph $(S,E,w)$ with arcs $E=\{(s_i,s_j)\mid i\neq j\}$ 
of weights $w(s_i,s_j)=\ov(s_i,s_j)$ (cf.\ Theorem 2.3 in~\cite{TaUk88}). 
In this setting, a path $\pi=s_{i_1}, \ldots, s_{i_k}$ corresponds to a string $\str(\pi):=\pr(s_{i_1},s_{i_2})\cdots \pr(s_{i_{k-1}},s_{i_k})s_{i_k}$.
By $\ov(\pi)$ we denote the total weight of arcs in the path $\pi$.

To accommodate reversals we extend the notion of an overlap graph to $\barG{S}=(V,E,w)$.
Here $V=\{v_s\,:\,s \in S\} \cup \{v^R_s\,:\,s \in S\}$ so every $s\in S$ corresponds to exactly two vertices,
$v_s$ and $v^R_s$. We define $\str(v_s)=s$ and $\str(v^R_s)=s^R$.
For a vertex $\alpha\in \barG{S}$ we define $\alpha^R$ as $v^R_s$ if $\alpha=v_s$ for some $s$
or as $v_s$ if $\alpha = v^R_s$ for some $s$. Note that $\str(\alpha^R)=\str(\alpha)^R$.
For every $\alpha,\beta \in V$, $\alpha \ne \beta$, we introduce an arc from $\alpha$ to $\beta$ with weight $\ov(\str(\alpha),\str(\beta))$.
For an arc $e=(\alpha,\beta)$ we define $e^R = (\beta^R,\alpha^R)$. Note that the weight of $e^R$ is the same as the weight of $e$.

For paths $\pi$ in $\barG{S}$ we also use the notions of $\str(\pi)$ and $\ov(\pi)$.
We say that a path $\pi$ in $\barG{S}$ is \emph{semi-Hamiltonian} if $\pi$ contains, for every vertex $\alpha\in \barG{S}$, exactly
one of the two vertices $\alpha$, $\alpha^R$.
Observe that a solution to SCS-R problem for a reverse-factor-free set $S$ corresponds to a maximum-weight semi-Hamiltonian path $\pi$ in the overlap graph $\barG{S}$.

\section{Greedy Algorithm and its Linear-Time Implementation}
We define an auxiliary procedure \MakeReverseFactorFree$(S)$ that
removes from $S$ all strings $u$ which are contained as a factor in $v$ or $v^R$ for some $v \in S$, $v \ne u$.
Note that the resulting set $S'$ is reverse-factor-free and, moreover, a string is a common superstring with reversals for $S'$ if and only if it is a common superstring with reversals for $S$.

\begin{example}
  Let $S=\{ab,aaa,aab,baa\}$. Then \MakeReverseFactorFree$(S)$ produces $S'=\{aaa,aab\}$ or $S'=\{aaa,baa\}$.
\end{example}

The Greedy-R algorithm works as follows:
while $|S|>1$, choose $u,v \in \barS$ (excluding $u=v$ and $u=v^{R}$) 
with largest overlap, insert into $S$ the string $u \otimes v$,
and remove from $S$ all strings among $u,v,u^{R},v^{R}$ that belong to $S$;
see the pseudocode.

\DontPrintSemicolon
\begin{algorithm}
  \KwSty{Algorithm} \textsl{Greedy-R}($S$)\\
  \KwIn{a non-empty set of strings $S$}
  \KwOut{a superstring of $S$ that approximates a solution of SCS-R problem for $S$}
  \Begin{
    $S$ := \MakeReverseFactorFree($S$) \;
    \While{$|S|>1$}{
      $P:=\{ (u,v) : u,v\in \barS, u\notin \{v,v^R\}\}$ \;
      $\{$ $S$ is reverse-factor-free and $|S|>1$, so $|P|\ge 1$ $\}$ \\
      take $(u,v)\in P$ with the maximal value of $\ov(u,v)$ \;
      $S:= S \cup \{u \otimes v$\} \;
      $S:=S \setminus \{u,v,u^{R},v^{R}\}$\; 
    }
    \Return{the only element of $S$}
  }
  \textbf{end}
\end{algorithm}

Let us state two properties of this algorithm useful for its efficient implementation.

\begin{lemma}\label{lem:free}
The set $S$ stays reverse-factor-free after each iteration of the \textbf{while} loop.
\end{lemma}
\begin{proof}
Suppose that at some point $S$ ceases to be reverse-factor-free.
This might only be due to the fact that, when $w=u\otimes v$ is introduced to $S$, $\barS$ contains a string
$w' \notin \{u,u^R,v,v^R\}$ such that $w'$ is a factor of $w$ but not of $u$ or $v$.
The latter, however, implies $\ov(u,w')>\ov(u,v)$ and $\ov(w',v)>\ov(u,v)$.
That contradicts the choice
of $(u,v)\in P$ maximizing $\ov(u,v)$.
\end{proof}

\begin{lemma}\label{lem:stays}
Before $u\otimes v$ is inserted to $S$, we have $\ov(w, u\otimes v)=\ov(w,u)$ and $\ov(u\otimes v,w)=\ov(v,w)$
for every $w\in \barS$.
\end{lemma}
\begin{proof}
Clearly, $\ov(w, u\otimes v)\ge \ov(w,u)$ and $\ov(u\otimes v,w)\ge \ov(v,w)$. 
Moreover, one of these inequalities might be strict only if $\ov(w, u\otimes v)>|u|$ or $\ov(u\otimes v, w)>|v|$, in particular, only if $w$ contains respectively $u$ or $v$ as a proper factor.
This, however, contradicts Lemma~\ref{lem:free}.
\end{proof}

\subsection{Interpretation on the Overlap Graph}

Tarhio and Ukkonen~\cite{TaUk88} work with the following interpretation of the greedy algorithm on the overlap graph
$G_S$: They maintain a set of arcs $F\subseteq E(G_S)$ forming a collection of disjoint paths and
at each step they add to $F$ an arc $e\in E(G_S)\setminus F$ of maximum weight so that the resulting set still forms a collection of disjoint paths.
Here, paths correspond to strings in the collection $S$ maintained by the original implementation. 
Insertion of an arc $(s_i,s_j)$ to $F$ results in merging two paths into one and this corresponds to replacing strings $u,v\in S$ with $u\otimes v$.
It turns out that $\ov(u,v)=\ov(s_i,s_j)$ and thus 
both implementations of the greedy algorithm are equivalent.

\DontPrintSemicolon
\begin{algorithm}[h]
  \KwSty{Algorithm} \textsl{Greedy-R2}($S$)\\
  \KwIn{a non-empty set of strings $S$}
  \KwOut{a superstring of $S$ that approximates a solution of SCS-R problem for $S$}
  \Begin{
    $S$ := \MakeReverseFactorFree($S$) \;
    Construct $\barG{S}$\;
    $F := \emptyset$\;
    \For{$i := 1$ \KwSty{to} $|S|-1$}{
      $P:=\{e \in E(\barG{S})\setminus F : F\cup\{e,e^R\}\text{ forms a collection of disjoint paths in }\barG{S}\}$ \;
      take $e\in P$ with the maximal weight $w(e)$ \;
      $F := F\cup\{e,e^R\}$\; 
    }
    \Return{$\str(\pi)$ for one of the maximal paths $\pi$ formed by $F$ in $\barG{S}$}
  }
  \textbf{end}
\end{algorithm}

Here, we provide an analogous interpretation of Greedy-R on the overlap graph $\barG{S}$
and, for completeness, explicitly prove its equivalence to the original implementation.
We also maintain a set of arcs $F\subseteq E(\barG{S})$ forming a 
collection of disjoint paths.  In each step we extend $F$ with a pair of arcs $\{e,e^R\}$ of maximum (common)
weight so that the resulting set still forms a collection of disjoint paths; see the pseudocode of algorithm Greedy-R2.
Observe that this way paths formed by $F$ come in pairs $\pi,\pi^R$ such that $\pi=\alpha_1,\ldots,\alpha_p$
and $\pi^R=\alpha_p^R,\ldots,\alpha_1^R$ (in particular, $\str(\pi^R)=\str(\pi)^R$).
We claim that these pairs of paths correspond to strings $s\in S$ of the original implementation (with $s=\str(\pi)$ or $s=\str(\pi^R)$).
In particular, each string $s\in \barS$ corresponds to a single path formed by $F$ unless
$s=s^R$ when it corresponds to two reverse paths. 
 The claim is certainly true at the beginning of the algorithm, so let us argue
that single iterations of the main loops in both algorithms perform analogous operations.

First, consider an arc $e=(\alpha,\beta)$ such that there exists a path $\pi$ that ends at $\alpha$ and a path $\pi'$
that starts at $\beta$. (Otherwise, $F\cup\{e,e^R\}$ does not form a collection of disjoint paths).
Observe that strings $u=\str(\pi)$ and $v=\str(\pi')$ belong to $\barS$ in Greedy-R and that $u\in \{v,v^R\}$
if and only if $F\cup\{e,e^R\}$ yields a cycle. Thus, both algorithms consider
essentially the same set of possibilities 
(the only caveat is that if $u$ or $v$ is a palindrome, then several arcs $e$ correspond to the same pair of strings from $\barS$). 
By Lemma~\ref{lem:stays} the weight of an arc $(\alpha,\beta)$ is $\ov(\str(\alpha),\str(\beta))=\ov(u,v)$. Thus, 
the maximum-weight arcs correspond to pairs $(u,v)\in P$ with largest overlap.

Clearly, setting $F:= F\cup\{e,e^R\}$ results in merging $\pi'$ with $\pi$
and $\pi^R$ with $\pi'^R$. These paths represent $u\otimes v$ and $v^R \otimes u^R$.
Since the set $S$ in Greedy-R stays reverse-factor-free due to Lemma~\ref{lem:free} and, in particular,
$S$ does not contain any pair of strings $v,v^R$ such that $v$ is not a palindrome,
the ``$S:=S\setminus \{u,u^R,v,v^R\}$'' instruction always removes exactly two elements of $S$.
This means that the bijection between $S$ and pairs of paths formed by $F$ is preserved.

Finally, observe that after exactly $|S|-1$ steps $F$ forms two semi-Hamiltonian paths.
Either of them can be returned as a solution.

\subsection{Linear-Time Implementation}

Ukkonen \cite{Uk90} showed that the Greedy algorithm for the original SCS problem
can be implemented in linear time based on the overlap-graph interpretation.
Our linear-time implementation of the Greedy-R2 algorithm is quite similar.

First, let us show how to efficiently implement the \MakeReverseFactorFree\ operation.
It is actually slightly easier to compute it for $\barS$ instead of $S$.
However, this is not an issue: if we substitute $S$ with $\barS$ in the very beginning of the greedy algorithm,
then the overlap graph will stay the same.

\begin{lemma}\label{lem:mrff}
  The result of \MakeReverseFactorFree$(\barS)$ can be computed in time linear in the total length of strings in $S$.
\end{lemma}

\begin{proof}
  Let us introduce an auxiliary procedure \MakeFactorFree$(X)$ 
  that removes from $X$ all strings $u$ for which there exists
  a string $v$ in $X$ such that $u$ is a proper factor of $v$.
  
    Ukkonen \cite{Uk90} applied the following result for the preprocessing phase of the greedy algorithm for the ordinary SCS problem.

  \begin{claim}[\cite{Uk90}]
    \MakeFactorFree$(X)$ can be implemented in time linear in the total length of the strings in $X$.
  \end{claim}
  
  Observe that in order to compute \MakeReverseFactorFree$(\barS)$ it suffices to determine
  $S' = \,$\MakeFactorFree$(\barS)$ and then for every pair of strings $u,u^R\in S'$
  leave exactly one of these strings, e.g., the lexicographically smaller one. 
  Note that $u\in S'$ if and only if $u^R \in S'$, so for the latter
  it suffices to iterate through $S'$ and report a string $u$ if and only if $u \le u^R$.
  
  The whole procedure works in linear time in the total length of strings in $\barS$,
  which is at most twice the total length of strings in $S$.
  \end{proof}

Now, we show the main result.

\begin{theorem}
  Greedy-R algorithm can be implemented in time linear in the total length of strings in $S$.
\end{theorem}
\begin{proof}
  By Lemma~\ref{lem:mrff}, we can make the input set reverse-factor-free in linear time.
  Let $S=\{s_1,\ldots,s_m\}$ be the set of remaining strings and $n$ the total length of strings in $S$.
  
  We actually implement the equivalent algorithm Greedy-R2.
  We cannot store the whole graph $\barG{S}$, since this would take too much space.
  Therefore we only store its vertex set, whereas we will be considering the edges of the graph in an indirect way.
  
  Denote by $\Pref(X)$ the set of all different prefixes of strings in $X$.
  Each element of this set can be represented as a state of the Aho-Corasick automaton constructed for $X$, thus
  using $\Oh(1)$ space per element.
  Further, given a string $w$, denote by $\PrefSet(w,X)$, $\SufSet(w,X)$ the sets (represented as lists of identifiers) of strings
  in $X$ having $w$ as a prefix, suffix, respectively.
  Ukkonen \cite{Uk90} applied the Aho-Corasick automaton for $X$ to show the following fact:

  \begin{claim}[\cite{Uk90}]
    $\PrefSet(w,X)$, $\SufSet(w,X)$ for all $w \in \Pref(X)$ can be computed in time
    linear in the total length of the strings in $X$.
  \end{claim}
  In the implementation we use sets of the form $\Pref(\barS)$, $\PrefSet(w,\barS)$ and $\SufSet(w,\barS)$.
  Our implementation actually requires $\PrefSet$ and $\SufSet$ sets to consist of vertices $\alpha\in V(\barG{S})$
  instead of strings $\str(\alpha)$. Thus, we replace every string identifier with one or two vertex identifiers
  if the string is not a palindrome or is a palindrome, respectively.

  Instead of simulating the main loop directly, we will consider all overlaps $w$ of $\str(\alpha)$ and $\str(\beta)$
  to find the maximum-weight arc $e=(\alpha,\beta)\in E(\barG{S})$.  
  Observe that at subsequent iterations of the loop, the weight $w(e)$ may only decrease.
  Hence, we iterate over all  $w \in \Pref(\barS)$ in decreasing length order
  and for each $w$ check if there exists an appropriate arc $e=(\alpha,\beta)$ such that $F\cup\{e,e^R\}$
  forms a collection of disjoint paths. More formally, we seek vertices $\alpha,\beta\in V(\barG{S})$ such that:
  \begin{enumerate}[(a)]
    \item a path formed by $F$ ends in $\alpha$ and a \emph{different} path formed by $F$ starts at $\beta$,
    \item $\beta \ne \alpha^R$, and
    \item $w$ is the overlap of $\str(\alpha)$ and $\str(\beta)$.
  \end{enumerate}
  Once we find such an arc $e$, we set $F:=F\cup\{e,e^R\}$. 

  Let us explain how this approach can be implemented efficiently.  
  To iterate through all $w \in \Pref(\barS)$ in decreasing length order we simply traverse all the states of the Aho-Corasick automaton
  in reverse-BFS order, breaking ties arbitrarily.
  To check the conditions (a)-(c), we could iterate through pairs of elements
  $\beta \in \PrefSet(w,\barS)$ and $\alpha \in \SufSet(w,\barS)$ and verify if they satisfy the conditions.
  However, to avoid repetitively scanning \emph{redundant} elements, we remove $\beta\in \PrefSet(w,\barS)$
  and $\alpha\in \SufSet(w,\barS)$ if no path starts in $\beta$ and no path ends in $\alpha$, respectively.
  For each vertex we will remember if a path starts or ends there, and, if so, what is the other endpoint of the path.
  Observe that for each $\alpha \in \SufSet(w,\barS)$ there are at most two non-redundant elements $\beta \in \PrefSet(w,\barS)$
  for which the arc $e=(\alpha,\beta)$ is not valid.
  Indeed, these might only be $\alpha^R$ and the starting vertex of the path ending at $\alpha$. 
  
  Hence, with amortized constant-time overhead, either an arc $e=(\alpha,\beta)$ satisfying (a)-(c) is found,
  or it can be verified that no such arc exists for this particular string $w$ and then
  we can continue iterating through strings in $\Pref(\barS)$.
  If an arc is found, we start the next search with the same string $w$, since
  there could still be arcs satisfying conditions (a)-(c) for the same string $w$.

  To conclude: in every step of the simulation, in amortized constant time we either find a pair of arcs to be introduced to $F$
  or discard the given candidate $w \in \Pref(\barS)$.
  The former situation takes place at most $m-1$ times and the latter happens at most $n$ times.
  The whole algorithm thus works in $\Oh(n)$ time.
 \end{proof}

\section{Compression Ratio}

In this section we prove that the compression ratio of the Greedy-R algorithm is always at least $\frac12$, and that this value is effectively achieved, so the bound is tight.

Let $S=\{s_1,\ldots,s_m\}$ be the input set of strings.
We assume that $S$ is already reverse-factor-free.
Let $\opt(S)$ be the length of a longest semi-Hamiltonian path in $G=\barG{S}$.
Let $\pgreedy(S)$ denote the length of the semi-Hamiltonian path produced by the Greedy-R algorithm for $S$.
We will show that $\pgreedy(S) \ge \frac12 \opt(S)$.

In the proof we use as a tool the following fact from \cite{TaUk88} (see Lemma~3.1 in \cite{TaUk88}):
\begin{lemma}\label{lem:4vert}
  If strings $x_1,x_2,x_3,x_4$ satisfy
  $$\max(\ov(x_1,x_4),\ov(x_2,x_3)) \le \ov(x_1,x_3),$$
  then
  $$\ov(x_1,x_4)+\ov(x_2,x_3) \le \ov(x_1,x_3)+\ov(x_2,x_4).$$
\end{lemma}

We proceed with the following crucial lemma.

\begin{lemma}\label{lem:compression_ratio}
  Let $S$ be a set of strings and let $u,v \in \barS$ be two elements for which $\ov(u,v)$ is maximal
  ($u \notin \{v,v^R\}$).
  Set $\OV=\ov(u,v)$.
  Let $\opt(S)$ be the length of a longest semi-Hamiltonian path in $G$
  and let $\opt'(S)$ be the length of a longest semi-Hamiltonian path in $G$ that contains the arc $(u,v)$.
  Then:
  $$\opt'(S) \ge \opt(S)-\OV.$$
\end{lemma}

\begin{proof}
  We consider the path $\pi$ corresponding to $\opt(S)$ and show how it can be modified
  without losing its semi-Hamiltonicity so that the arc $(u,v)$ occurs in the path and the length of the path decreases by at most $\OV$.

  Obviously, if $\pi$ already contains the arc $(u,v)$, nothing is to be done.
  If both $u$ and $v$ occur in $\pi$, we perform transformations as in the proof of a similar fact
  from \cite{TaUk88} related to the ordinary SCS problem.
  If $u$ occurs in $\pi$ before $v$ then we select $\pi'$ as in Fig.~\ref{fig:case_a} and:
  $$\ov(\pi') \ge \ov(\pi)-\ov(u,b)-\ov(c,v)+\ov(u,v) \ge \ov(\pi)-\ov(u,b) \ge \ov(\pi)-\OV.$$
  Note that both nodes $b, c$ exist (they could be the same node, though).
  If any of the remaining nodes does not exist, it is simply skipped on the path.

  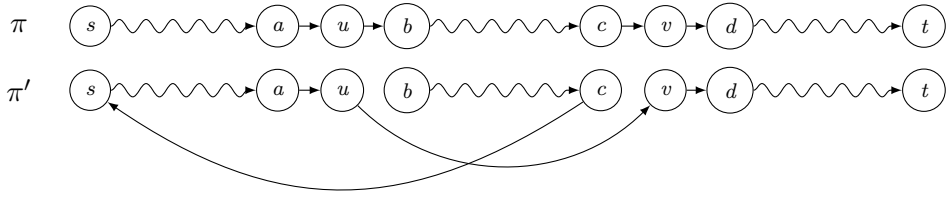
\begin{figure}[htpb]
  \begin{center}
      \tikzstyle{mynode} = [circle,draw]
      \begin{tikzpicture}[scale=0.85]
        \draw (-1,0) node {$\pi$};
        \foreach \x/\c in {0.1/s,3/a,4/u,5/b,8/c,9/v,10/d,13/t}
          \draw (\x,0) node[mynode] (\c) {\it\scriptsize \c};
        \foreach \x/\y in {s/a,b/c,d/t}
          \draw[decorate, decoration={snake},-latex] (\x) -- (\y);
        \foreach \x/\y in {a/u,u/b,c/v,v/d}
          \draw[-latex] (\x) -- (\y);

        \begin{scope}[yshift=-1cm]
        \draw (-1,0) node {$\pi'$};
        \foreach \x/\c in {0.1/s,3/a,4/u,5/b,8/c,9/v,10/d,13/t}
          \draw (\x,0) node[mynode] (\c) {\it\scriptsize \c};
        \foreach \x/\y in {s/a,b/c,d/t}
          \draw[decorate, decoration={snake},-latex] (\x) -- (\y);
        \foreach \x/\y in {a/u,v/d}
          \draw[-latex] (\x) -- (\y);
        \draw[-latex] (c) .. controls (5,-2) and (3,-2) .. (s); %
        \draw[-latex] (u) .. controls (5.5,-1.5) and (7.5,-1.5) .. (v); %
        \end{scope}
      \end{tikzpicture}
  \end{center}
  \vspace*{-0.4cm}
  \caption{\label{fig:case_a}Proof of Lemma~\ref{lem:compression_ratio}, first case: $u$ occurs in $\pi$ before $v$.}
\end{figure}

  If $v$ occurs before $u$, then by applying the inequality of Lemma~\ref{lem:4vert}
  (with $x_1=u$, $x_2=a$, $x_3=v$, $x_4=d$)
  we have 
  $$\ov(u,d)+\ov(a,v) \le \ov(u,v)+\ov(a,d).$$
  By this inequality, for the path $\pi'$ defined in Fig.~\ref{fig:case_b} we have:
  \begin{align*}
    \ov(\pi') &\ge \ov(\pi)-\ov(a,v)-\ov(u,d)+\ov(u,v)+\ov(a,d)-\ov(v,b)\\
              &\ge \ov(\pi)-\ov(v,b) \ge \ov(\pi)-\OV.
  \end{align*}
  As before, if any of the depicted nodes does not exist, we simply skip the corresponding part of the path.
  In particular, if any of the nodes $u,v$ is an endpoint of the path $\pi$, we do not need to use the aforementioned
  inequality to show that $\ov(\pi') \ge \ov(\pi)-\OV$.

  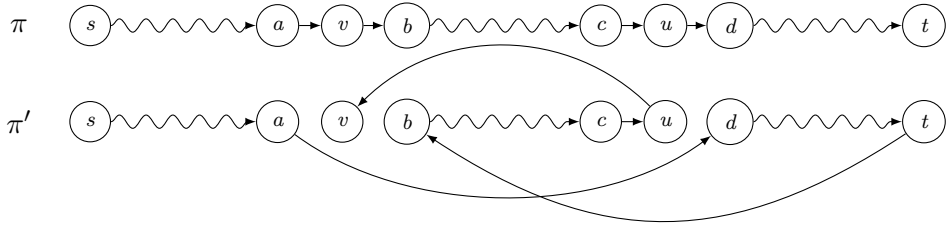
\begin{figure}[htpb]
  \begin{center}
      \tikzstyle{mynode} = [circle,draw]
      \begin{tikzpicture}[scale=0.85]
        \draw (-1,0) node {$\pi$};
        \foreach \x/\c in {0.1/s,3/a,4/v,5/b,8/c,9/u,10/d,13/t}
          \draw (\x,0) node[mynode] (\c) {\it\scriptsize \c};
        \foreach \x/\y in {s/a,b/c,d/t}
          \draw[decorate, decoration={snake},-latex] (\x) -- (\y);
        \foreach \x/\y in {a/v,v/b,c/u,u/d}
          \draw[-latex] (\x) -- (\y);

        \begin{scope}[yshift=-1.5cm]
        \draw (-1,0) node {$\pi'$};
        \foreach \x/\c in {0.1/s,3/a,4/v,5/b,8/c,9/u,10/d,13/t}
          \draw (\x,0) node[mynode] (\c) {\it\scriptsize \c};
        \foreach \x/\y in {s/a,b/c,d/t}
          \draw[decorate, decoration={snake},-latex] (\x) -- (\y);
        \foreach \x/\y in {c/u}
          \draw[-latex] (\x) -- (\y);
        \draw[-latex] (a) .. controls (5,-1.5) and (8,-1.5) .. (d); %
        \draw[-latex] (t) .. controls (10,-2) and (8,-2) .. (b); %
        \draw[-latex] (u) .. controls (7.5,1.5) and (5.5,1.5) .. (v); %
        \end{scope}
      \end{tikzpicture}
  \end{center}
  \vspace*{-0.4cm}
  \caption{\label{fig:case_b}Proof of Lemma~\ref{lem:compression_ratio}, second case: $u$ occurs in $\pi$ after $v$.}
\end{figure}

  Differently from the original SCS problem considered in \cite{TaUk88}, it
  might not be the case that $u$ and $v$ are in $\pi$.
  If none of them is, then $\pi$ contains both $u^R$ and $v^R$ and by reversing $\pi$ (that is, taking the path $\pi^R$)
  we obtain a semi-Hamiltonian path that contains both $u$ and $v$,
  which was the case considered before.
  Thus, we can assume that $u$ and $v^{R}$ occur in $\pi$ (the case of $u^{R}$ and $v$ is symmetric).
  Again, we have two cases, depending on which of the two nodes comes first in $\pi$.
  If $u$ occurs before $v^{R}$ in $\pi$, then we have (note that $\ov(c,v^{R})=\ov(v,c^{R})$):
  $$\ov(\pi') \ge \ov(\pi)-\ov(u,b)-\ov(v^{R},d)+\ov(u,v) \ge \ov(\pi)-\ov(u,b) \ge \ov(\pi)-\OV,$$
  see also Fig.~\ref{fig:case_c}.
 
Finally, if $v^{R}$ occurs before $u$, then (see Fig.~\ref{fig:case_d}):
  $$\ov(\pi') \ge \ov(\pi)-\ov(v^{R},b)-\ov(u,d)+\ov(u,v) \ge \ov(\pi)-\ov(v^{R},b) \ge \ov(\pi)-\OV.$$

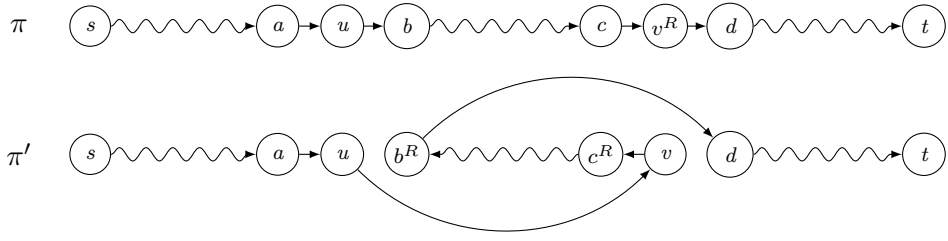
\begin{figure}
  \begin{center}
      \tikzstyle{mynode} = [circle,draw]
      \begin{tikzpicture}[scale=0.85]
        \draw (-1,0) node {$\pi$};
        \draw (9,0) node[mynode,inner sep=0pt,minimum size=0.58cm] (vr) {\scriptsize $v^{R}$};
        \foreach \x/\c in {0.1/s,3/a,4/u,5/b,8/c,10/d,13/t}
          \draw (\x,0) node[mynode] (\c) {\it\scriptsize \c};
        \foreach \x/\y in {s/a,b/c,d/t}
          \draw[decorate, decoration={snake},-latex] (\x) -- (\y);
        \foreach \x/\y in {a/u,u/b,c/vr,vr/d}
          \draw[-latex] (\x) -- (\y);

        \begin{scope}[yshift=-2cm]
        \draw (-1,0) node {$\pi'$};
        \draw (5,0) node[mynode,inner sep=0pt,minimum size=0.58cm] (br) {\scriptsize $b^{R}$};
        \draw (8,0) node[mynode,inner sep=0pt,minimum size=0.58cm] (cr) {\scriptsize $c^{R}$};
        \foreach \x/\c in {0.1/s,3/a,4/u,9/v,10/d,13/t}
          \draw (\x,0) node[mynode] (\c) {\it\scriptsize \c};
        \foreach \x/\y in {s/a,cr/br,d/t}
          \draw[decorate, decoration={snake},-latex] (\x) -- (\y);
        \foreach \x/\y in {a/u,v/cr}
          \draw[-latex] (\x) -- (\y);
        \draw[-latex] (u) .. controls (5.5,-1.5) and (7.5,-1.5) .. (v); %
        \draw[-latex] (br) .. controls (6.5,1.5) and (8.5,1.5) .. (d); %
        \end{scope}
      \end{tikzpicture}
  \end{center}
  \vspace*{-0.4cm}
  \caption{\label{fig:case_c}Proof of Lemma~\ref{lem:compression_ratio}, third case: $u$ occurs in $\pi$ before $v^R$.}
\end{figure}

\begin{figure}
\bigskip
  \begin{center}
      \tikzstyle{mynode} = [circle,draw]
      \begin{tikzpicture}[scale=0.85]
        \draw (-1,0) node {$\pi$};
        \draw (4,0) node[mynode,inner sep=0pt,minimum size=0.56cm] (vr) {\scriptsize $v^{R}$};
        \foreach \x/\c in {0.1/s,3/a,5/b,8/c,9/u,10/d,13/t}
          \draw (\x,0) node[mynode] (\c) {\it\scriptsize \c};
        \foreach \x/\y in {s/a,b/c,d/t}
          \draw[decorate, decoration={snake},-latex] (\x) -- (\y);
        \foreach \x/\y in {a/vr,vr/b,c/u,u/d}
          \draw[-latex] (\x) -- (\y);

        \begin{scope}[yshift=-2cm]
        \draw (-1,0) node {$\pi'$};
        \draw (0.1,0) node[mynode,inner sep=0pt,minimum size=0.56cm] (sr) {\scriptsize $s^{R}$};
        \draw (3,0) node[mynode,inner sep=0pt,minimum size=0.56cm] (ar) {\scriptsize $a^{R}$};
        \foreach \x/\c in {4/v,5/b,8/c,9/u,10/d,13/t}
          \draw (\x,0) node[mynode] (\c) {\it\scriptsize \c};
        \foreach \x/\y in {ar/sr,b/c,d/t}
          \draw[decorate, decoration={snake},-latex] (\x) -- (\y);
        \foreach \x/\y in {v/ar,c/u}
          \draw[-latex] (\x) -- (\y);
        \draw[-latex] (u) .. controls (7.5,1.5) and (5.5,1.5) .. (v); %
        \draw[-latex] (sr) .. controls (3,-2) and (7,-2) .. (d); %
        \end{scope}
      \end{tikzpicture}
  \end{center}
  \vspace*{-0.4cm}
  \caption{\label{fig:case_d}Proof of Lemma~\ref{lem:compression_ratio}, third case: $u$ occurs in $\pi$ after $v^R$.}
\end{figure}
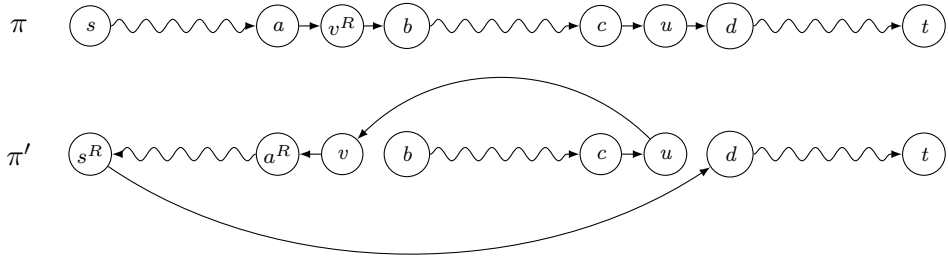

This completes the proof of the lemma.
 \end{proof}

To conclude the proof of the compression ratio of the Greedy-R algorithm we use an inductive argument.
Assume that $\pgreedy(S') \ge \frac12 \opt(S')$ for all $S'$ such that $|S'|<|S|$.
By Lemma~\ref{lem:compression_ratio}, for $S'=S \setminus \{u,v,u^R,v^R\} \cup \{u \otimes v\}$:
$$\opt(S)-\OV \le \opt'(S) \le \opt(S')+\OV,$$
so that $\opt(S')+2\OV \ge \opt(S)$.
Moreover, $\pgreedy(S)=\pgreedy(S')+\OV$ and, by the inductive hypothesis, $\pgreedy(S') \ge \frac12 \opt(S')$.
Consequently:
$$\pgreedy(S) = \pgreedy(S')+\OV \ge \frac12 (\opt(S')+2\OV) \ge \frac12 \opt(S).$$
We arrive at the following theorem.

\begin{theorem}
  Greedy-R has compression ratio $\frac12$.
\end{theorem}

It turns out that the bound on the compression ratio of the Greedy-R algorithm is tight.
It suffices to consider the set of strings $\{ab^h,b^hc,b^{h+1}\}$
(this is actually the same example as from the analysis of the Greedy algorithm \cite{BlJiLiTrYa94}).
The output of Greedy-R can be the string $ab^hcb^{h+1}$ with total overlap $h$,
whereas an optimal solution to SCS-R is $ab^{h+1}c$ of total overlap $2h$.

\section{NP-completeness of the SCS-R Problem}

  Let $\Gamma=\Sigma\cup\{\$,\#\}$, where $ \$,\# \not \in \Sigma$.
  Let $h$ be the morphism from $\Sigma^*$ to $\Gamma^*$  defined by $h(c)=\$\#\,c$, for every $c \in \Sigma$.
  Given $k\geq 1$, we also define the morphism $g_k(c)=c^k$, for every $c \in \Sigma$. 

  \begin{observation}\label{obs:small_ov}
    For every nonempty strings $u,v$, the strings $h(u)$ and $h(v)^{R}$ have an overlap of length at most one.
  \end{observation}

  \begin{observation}\label{obs:multiple}
    For every nonempty strings $u,v$, one has
    $$|h(g_k(u))|=3k|u| \quad\mbox{and }\quad \ov(h(g_k(u)),h(g_k(v))) = 3k\cdot\ov(u,v).$$
  \end{observation}

  In the decision version of the SCS-R problem we need to check if a given set of strings
  admits a common superstring with reversals of length at most $\ell$, where $\ell$ is
  an additional input parameter.

  \begin{proposition}\label{NP}
    The decision version of the SCS-R problem is NP-complete.
  \end{proposition}
  \begin{proof}
    Let $s_1,\ldots,s_m$ be an instance of the SCS problem.
    Set  $y_i=h(g_m(s_i))$, for $i=1,\ldots,m$.
    We have the following claim.

    \begin{claim}
      There exists a common superstring of $s_1,\ldots,s_m$ of length at most $\ell$
      if and only if
      there exists a common superstring with reversal of $y_1,\ldots,y_m$ of length at most $3m\ell$.
    \end{claim}
    \begin{proof}
      $(\Rightarrow)$
      If $u$ is a common superstring of $s_1,\ldots,s_m$,
      then $h(g_m(u))$ is a common superstring of $y_1,\ldots,y_m$.
      Hence, $h(g_m(u))$ is also a common superstring with reversals of $y_1,\ldots,y_m$.
      If $|u|=\ell$, then $|h(g_m(u))|=3m\ell$.

      $(\Leftarrow)$
      Let $u$ be a common superstring with reversals of $y_1,\ldots,y_m$.
      We will show that, thanks to the special form of the strings, there exists
      a common superstring \emph{without} reversals of $y_1,\ldots,y_m$ of length not much greater than $|u|$.

      Let $\pi=z_{i_1},\ldots,z_{i_m}$ be the sequence of nodes on the path in the overlap graph $G$
      that corresponds to $u$;
      here $\{i_1,\ldots,i_m\}=\{1,\ldots,m\}$ and each $z_i$ is either $y_i$ or $y^R_i$. Let us construct a new sequence of nodes $\pi'$, that first contains all nodes from $\pi$
      of the form $y_i$ in the same order as in $\pi$, and then all nodes from $\pi$
      of the form $y_i^R$, but given in the reverse order and taken without reversal.
      Let $u'$ be the common superstring corresponding to $\pi'$.
      By Observation~\ref{obs:small_ov}, $|u'| \le |u|+m-1$.
      Note that $u'$ is an ordinary common superstring of $y_1,\ldots,y_m$.

      By Observation~\ref{obs:multiple}, $|u'|$ is a multiple of $3m$ and $u'$
      corresponds to a common superstring $v$ of $s_1,\ldots,s_m$
      of length $|u'|/(3m)$.
      If $|u|\le 3m\ell$, then
      $$|v|=\frac{|u'|}{3m} \le \left\lfloor\frac{|u|+m-1}{3m}\right\rfloor \le \ell.$$
    
    \vspace{-0.5cm}
    \end{proof}
    
    The claim provides a reduction of the decision version of the SCS problem to the decision version of the SCS-R problem.
    This shows that the latter is NP-hard, hence NP-complete, as it is obviously in NP.
 \end{proof}

\section{Acknowledgments}

The authors thank anonymous referees for a number of helpful comments and remarks.

This work started during a visit of Gabriele Fici to the University of Warsaw funded by the Warsaw Center of Mathematics and Computer Science. 

Gabriele Fici is supported by the PRIN 2010/2011 project ``Automi e Linguaggi Formali: Aspetti Matematici e Applicativi'' of the Italian Ministry of Education (MIUR).
Tomasz Kociumaka is supported by Polish budget funds for science in 2013-2017 as a research project under the `Diamond Grant' program.
Jakub Radoszewski, Wojciech Rytter and Tomasz Wale\'n are supported by the Polish National Science Center, grant no 2014/13/B/ST6/00770.


\bibliographystyle{plain}

\end{document}